\newcommand{\version}{October 17, 2017}
\documentclass[11pt,a4paper,reqno]{amsart}
\usepackage{amsthm,amsmath,amsfonts,amssymb,amsxtra,appendix,bookmark,dsfont,latexsym}

\makeatletter
\usepackage{hyperref}
\usepackage{delarray,color}
\usepackage{euscript}

\usepackage{hyperref}

\makeatletter
\def\@tocline#1#2#3#4#5#6#7{\relax
  \ifnum #1>\c@tocdepth 
  \else
    \par \addpenalty\@secpenalty\addvspace{#2}%
    \begingroup \hyphenpenalty\@M
    \@ifempty{#4}{%
      \@tempdima\csname r@tocindent\number#1\endcsname\relax
    }{%
      \@tempdima#4\relax
    }%
    \parindent\z@ \leftskip#3\relax \advance\leftskip\@tempdima\relax
    \rightskip\@pnumwidth plus4em \parfillskip-\@pnumwidth
    #5\leavevmode\hskip-\@tempdima
      \ifcase #1
       \or\or \hskip 1em \or \hskip 2em \else \hskip 3em \fi%
      #6\nobreak\relax
      \dotfill
      \hbox to\@pnumwidth{\@tocpagenum{#7}}
    \par
    \nobreak
    \endgroup
  \fi}
\makeatother

\newcommand{\bdm}{\begin{displaymath}}
\newcommand{\edm}{\end{displaymath}}
\newcommand{\bdn}{\begin{eqnarray}}
\newcommand{\edn}{\end{eqnarray}}
\newcommand{\bay}{\begin{array}{c}}
\newcommand{\eay}{\end{array}}
\newcommand{\ben}{\begin{enumerate}}
\newcommand{\een}{\end{enumerate}}

\newcommand{\R}{\mathbb{R}}

\newcommand{\eps}{\varepsilon}

\newcommand{\half}{\hbox{$\frac12$}}

\newcommand{\intR}{\int_{\R ^2}}

\newcommand{\PsiLau}{\Psi_{\rm Lau}}
\newcommand{\cLau}{c _{\rm Lau}}

\newtheorem{theorem}{Theorem}[section]
\newtheorem{lemma}[theorem]{Lemma}
\newtheorem{asumption}[theorem]{Asumption}
\newtheorem{corollary}[theorem]{Corollary}
\newtheorem{proposition}[theorem]{Proposition}

\theoremstyle{definition}

\theoremstyle{remark}

\newcommand{\beq}{\begin{equation}}
\newcommand{\eeq}{\end{equation}}

\newcommand{\ebt}{E ^{\mathrm{bt}}}
\newcommand{\rhobt}{\rho ^{\mathrm{bt}}}
\newcommand{\norm}[1]{\left\lVert #1 \right\rVert}

\newcommand{\Eel}{E^{\rm el}}
\newcommand{\cEel}{\mathcal{E}^{\rm el}}
\newcommand{\rhoel}{\rho^{\rm el}}

\newcommand{\rhoMF}{\rho^{\rm MF}}
\newcommand{\EMF}{F^{\rm MF}}
\newcommand{\cEMF}{\mathcal{F}^{\rm MF}}

\newcommand{\Eelt}{\tilde{E}^{\rm el}}
\newcommand{\cEelt}{\tilde{\mathcal{E}}^{\rm el}}
\newcommand{\rhoelt}{\tilde{\rho}^{\rm el}}

\newcommand{\chiin}{\chi_{\mathrm{in}}}
\newcommand{\chiout}{\chi_{\mathrm{out}}}

\numberwithin{equation}{section}

\begin{document}

\title{The Laughlin liquid in an external potential}

\author{Nicolas Rougerie}
\address{Universit\'e Grenoble Alpes \& CNRS, LPMMC (UMR 5493), B.P. 166, F-38042 Grenoble, France}
\email{nicolas.rougerie@lpmmc.cnrs.fr}


\author{Jakob Yngvason}
\address{Faculty of Physics, University of Vienna, Boltzmanngasse 5, A-1090 Vienna, Austria}
\email{jakob.yngvason@univie.ac.at}



\date{\version}

\begin{abstract}
We study natural perturbations of the Laughlin state arising from the effects of trapping and disorder.  These are $N$-particle wave functions that have the form of a product of Laughlin states and analytic functions of the $N$ variables. We derive an upper bound to the ground state energy in a confining external potential, matching exactly a recently derived lower bound in the large $N$ limit. Irrespective of the shape of the confining potential, this sharp upper bound can be achieved through a modification of the  Laughlin function by suitably arranged quasi-holes. 
\end{abstract}

\maketitle

\tableofcontents

\section{Introduction}

In the physics of the fractional quantum Hall effect (FQHE)~\cite{StoTsuGos-99, Laughlin-99, Girvin-04,Jain-07} the Laughlin wave functions~\cite{Laughlin-83,Laughlin-87} play a pivotal role. Using complex coordinates $z_i$, $i=1,\dots, N$ for the positions of the particles in two-dimensional space and taking the magnetic length to be $1/\sqrt 2$, the function with filling factor $1/\ell$, where $\ell$ is a positive integer, is defined as
\beq\label{laufunc}
\PsiLau(z_1,\dots,z_N) =\cLau \prod_{i<j}(z_i-z_j)^{\ell} e^{-\sum_{i=1}^N|z_i|^2/2}
\eeq
with a normalization constant $\cLau$. For fermions $\ell$ takes odd values $\geq 3$ ($\ell=1$ corresponds to noninteracting particles) while for bosons even values are required. 

The state~\eqref{laufunc} was introduced as an ansatz for the ground state of the many-body quantum mechanical Hamiltonian (in complex notation)
\begin{equation}\label{eq:mag hamil}
H_N^{\rm QM} = \sum_{j=1} ^N \left( - 4 \partial_{z_j}\partial_{\bar{z}_j} - 2 z_j\partial_{z_j} + 2 \bar{z}_j\partial_{\bar{z}_j} + |z_j| ^2 + V (z_j) \right)+ \sum_{1\leq i < j \leq N} w(z_i-z_j) 
\end{equation}
acting on $L ^2 (\R ^{2N}) $, the Hilbert space for $N$ particles living in 2D. The one-particle part of the above contains the magnetic Laplacian $- 4 \partial_{z}\partial_{\bar{z}} - 2 z\partial_{z} + 2 \bar{z}\partial_{\bar{z}} + |z| ^2$ with constant magnetic field\footnote{We have chosen units so that the strength of the magnetic field is $2$ and the length unit is the magnetic length, $1/\sqrt 2$. Also, $\hbar=1$ and the mass is $\half$.} perpendicular (pointing downwards) to the plane where the particles move. In the ansatz~\eqref{laufunc}, all particles  live in the ground eigenspace (lowest Landau level, LLL in the sequel) of this magnetic Laplacian, in order to minimize the magnetic kinetic energy.  The correlation factors $(z_i-z_j)^\ell$ are inserted to suppress repulsive interactions $w(z_i-z_j)$ between the particles. 

The most general wave function in the LLL retaining the correlations of~\eqref{laufunc} is
\begin{equation}\label{eq:fullcorr} 
\Psi_F (z_1,\dots, z_N)= F(z_1,\dots, z_N)\PsiLau (z_1,\dots, z_N)
\end{equation}
with $F$ analytic and symmetric under exchange of the $z_i$. If the external potential in~\eqref{eq:mag hamil} is neglected, $V=0$, and for strong repulsive interactions $w$, all states of the form \eqref{eq:fullcorr} can in first approximation be expected to minimize the Hamiltonian~\eqref{eq:mag hamil}. Since $\Psi_F$ fully resides in the LLL, this only amounts to assuming that the repulsive interactions are rendered negligible by the pair correlations included in~\eqref{laufunc}. This is even exactly fulfilled for some model interaction operators~\cite{Haldane-83,TruKiv-85,PapBer-01} where all ground states are of the form \eqref{eq:fullcorr}.

An important issue, however, is to consider the effects of trapping and disorder. When the magnetic field and the interaction set the largest energy scales of the problem, this amounts to minimizing the potential energy in an external potential $V$ within the class of wave-functions~\eqref{eq:fullcorr}. That is, we look for analytic functions $F$ which optimize the energy 
\begin{equation}\label{eq:start ener}
\left\langle \Psi_F \, \big|\, \sum_{j=1} ^N V (z_j) \,\big|\, \Psi_F \right\rangle = \intR \rho_F (z) V(z) dz 
\end{equation}
where $V:\R^2 \mapsto \R$ is the external potential modeling trapping/disorder and 
\begin{equation}\label{eq:original density}
\rho_F (z) = N \int_{\R ^{2(N-1)}}  \left| \Psi_F \left(z, z_2,\ldots,z_N \right) \right| ^2 dz_2 \ldots dz_N
\end{equation}
is the one-particle density of the wave function~\eqref{eq:fullcorr}, normalized so as to have total mass~$N$. Defining
\begin{equation}\label{eq:min ener}
 E_V (N,\ell) = \inf \left\{ \intR \rho_F (z) V(z) dz \: | \: \Psi_F \mbox{ of the form~\eqref{eq:fullcorr} }, \, \int_{\R ^{2N}} |\Psi_F| ^2 = 1 \right\}, 
\end{equation}
an educated guess, supported by the analysis in \cite{RouSerYng-13a, RouSerYng-13b, LieRouYng-16, LieRouYng-17},  is that for large particle numbers $N$ with $\ell$ fixed,
\begin{equation}\label{eq:formal result}
\boxed{E_V (N,\ell) \simeq \ebt_V (N,\ell)}  
\end{equation}
where the \emph{bathtub energy} $\ebt_V (N,\ell)$ is defined as the lowest possible energy for normalized densities satisfying the bound $0 \leq \rho \leq (\pi\ell)^{-1}$:
\begin{equation}\label{eq:bathtub intro}
\ebt_V (N,\ell) := \inf \left\{ \intR \rho (z) V(z) dz \: | \: 0 \leq \rho \leq \frac{1}{\pi \ell},\: \intR \rho = N  \right\} = \intR \rhobt_V (z) V (z) dz. 
\end{equation}
It is well-known~\cite[Theorem~1.14]{LieLos-01} that its minimizers (we denote them by $\rhobt_V$ ) are found by saturating the upper bound for the density  and filling the level sets of $V$ up to a certain level fixed by the normalization. Thus  a bathtub minimizer  is uniquely determined by its support, which has area $N(\pi\ell)$, and where it takes the constant value $(\pi\ell)^{-1}$.

\medskip

The asymptotics~\eqref{eq:formal result} relies on two complementary properties of the class of states \eqref{eq:fullcorr}:

\medskip

\noindent\textbf{(a)} The Laughlin liquid is incompressible. In particular, any wave-function~\eqref{eq:fullcorr} based on the Laughlin state~\eqref{laufunc} has its one-body density everywhere bounded above by $1/(\pi \ell)$, which is precisely the particle density for the Laughlin function itself  within the disk where it is essentially supported. Given this basic density bound it is natural to expect the lower bound  $E_V (N,\ell) \gtrapprox \ebt_V (N,\ell)$ to hold.

\medskip

\noindent\textbf{(b)} The variational set of functions~\eqref{eq:fullcorr} is sufficiently large, so that one can construct a trial state whose density distribution (asymptotically) matches that of the bathtub minimizer. This leads to the upper bound $E_V (N,\ell) \lessapprox \ebt_V (N,\ell)$.

\medskip

Property \textbf{(a)} above has been rigorously established in~\cite{LieRouYng-16,LieRouYng-17}, following previous results in~\cite{RouYng-14,RouYng-15}. It was proved that the one-body density~\eqref{eq:original density} satisfies, in a suitable average sense, a universal bound  for all $F$,
\beq\label{eq:densbound}
\rho_F (z)\leq (\pi\ell)^{-1},
\eeq
in the limit of large particle numbers\footnote{For finite $N$ the density may oscillate close to the edge of the sample and exceed the bound \eqref{eq:densbound}, cf. the numerical calculations for $N=400$ in \cite{Ciftja-06}.} $N\to\infty$. 

Since the right-hand side of~\eqref{eq:densbound} is the particle density for the Laughlin function, the bound~\eqref{eq:densbound} means that an additional factor $F$ cannot compress the density beyond this limit to take full advantage of the minima of an external potential. This highly nontrivial result is an important aspect of the rigidity of the Laughlin liquid with respect to external perturbations. It relies essentially on the analyticity of $F$ which, in turn, is due to the restriction to the LLL. It is in stark contrast with the fact that, without a strong magnetic field, the electron density in a crystal can be arbitrarily high locally due to constructive interference of Bloch waves, each of which is uniformly bounded.

Property \textbf{(b)} is the concern of the present paper. Something more precise can in fact be shown, namely it is sufficient to consider the sub-class of wave-functions 
\begin{equation}\label{eq:fullcorr qh}
\Psi_f (z_1,\ldots,z_N) = \prod_{j=1} ^N f (z_j) \PsiLau (z_1,\dots, z_N) 
\end{equation}
with $f$ a polynomial in a {\em single} variable. Denoting the corresponding one-particle density by~$\rho_f$ the  minimal energy within this class is
\begin{equation}\label{eq:min ener e}
 e_V (N,\ell) = \inf \left\{ \intR \rho_f (z) V(z) dz \: | \: \Psi_f \mbox{ of the form~\eqref{eq:fullcorr qh} }, \, \int_{\R ^{2N}} |\Psi_f| ^2 = 1 \right\}. 
\end{equation}
Clearly, 
\begin{equation}\label{eq:trivial}
 E_V (N,\ell) \leq e_V (N,\ell) 
\end{equation}
for  we have simply reduced the variational set. A function of the above form 
corresponds to inserting Laughlin quasi-holes~\cite{Laughlin-83,Laughlin-87} with locations 
at the zeros of the polynomial $f$.
Each of those carries a fraction $1/\ell$ of an electron's charge, and is expected to behave as an anyon~\cite{AroSchWil-84,LunRou-16} with statistics parameter $-1/\ell$.

In this paper we complete the proof of~\eqref{eq:formal result} by showing that 
$$ e_V (N,\ell) \lessapprox \ebt_V (N,\ell).$$
In particular,  we refine Property~\textbf{(b)} by showing that suitable states whose density asymptotically minimizes the bathtub energy can always be found among those of the form~\eqref{eq:fullcorr qh} with polynomials $f$. In other words, not only does the density of an optimizer of~\eqref{eq:start ener} always reduce to that of the bathtub problem~\eqref{eq:bathtub intro}, but also it can be approximated by inserting {\em uncorrelated} Laughlin quasi-holes on top of the Laughlin wave-function.

The remarkable fact here is that no electron/quasi-hole correlations are needed to optimize the energy. The electron/electron correlations are sufficient to deal with all physical effects of the interaction and are robust against perturbations by external potentials. 
The reason why this is remarkable is that the distribution of (a-priori correlated) quasi-holes governed by $F(z_1,\ldots,z_N)$ actually sees a complicated many-body Hamiltonian, encoded in the factor $\PsiLau$ it gets multiplied with before entering the minimization problem~\eqref{eq:start ener}.

It was one of the key guesses of Laughlin's original theory~\cite{Laughlin-83,Laughlin-87} that the response of his wave-functions to external potentials could be investigated by inserting uncorrelated quasi-particles on top of it. In this paper we provide a complete mathematical proof of this fact, the first to our knowledge.

\medskip

\noindent{\bf Acknowledgments.} We thank Elliott H. Lieb for helpful remarks. N. Rougerie received financial support from the French ANR project ANR-13-JS01-0005-01.

\section{Statements}

We now turn to the precise statements of our main results, starting with our assumptions on the external potential. The density of the Laughlin state  is essentially supported in a thermodynamically large region, namely it is a droplet of radius $\propto \sqrt{N}$, and the confinement should keep the perturbed state also in a region of this order of magnitude.  Thus,  assumptions on the potential are most conveniently stated in terms of a scaled version of an $N$-dependent $V$. The simplest way is to write 
\begin{equation}\label{eq:asum pot 1}
V (z) = U \left(\frac{z}{\sqrt N} \right)  
\end{equation}
where $U$ is a \emph{fixed} function\footnote{Generalizations of this assumption are discussed in Section 5} satisfying the following conditions:

\begin{asumption}[\textbf{The external potential}]\label{asum:pot}\mbox{}\\
The scaled potential $U$ is a fixed, twice continuously differentiable function from $\R^2$ to $\R^+$. We moreover assume that 
$$ U(x) \underset{|x| \to \infty}{\longrightarrow} +\infty$$
but with at most polynomial growth of $U$ and its gradient: There exists fixed positive numbers $s,t$ such that 
\begin{equation}\label{eq:growth pot}
 |U(x)| \leq |x| ^s \mbox{ and } |\nabla U (x)| \leq |x| ^t \mbox{ for } |x| \mbox{ large enough.} 
\end{equation}
Furthermore, we assume that $U$ has no flat pieces: The Lebesgue measure of the level set $\left\{ U = e \right\}$ is zero for any $e\in \R ^+$.
\end{asumption}

That $U$ takes positive values is just a 
convention on the energy reference, ensuring in particular that the relevant energies are $O(N)$. Moreover, we assume regularity and a reasonable trapping behavior. That the potential has no flat pieces is just a convenient technical assumption, ensuring in particular that the ground state of the bathtub problem has a unique solution. This allows a simple statement about convergence of densities (Corollary \ref{cor:density} below).
\medskip

Our main result is as follows: 
\begin{theorem}[\textbf{Potential energy of the Laughlin liquid}]\label{thm:main}\mbox{}\\
For fixed integer $\ell$ we have,  under Assumption~\ref{asum:pot},
\begin{equation}
 \label{eq:main result bis}
\lim_{N\to\infty}\frac {E_V(\ell,N)}{ e_V(\ell,N)}=\lim_{N\to\infty}\frac {E_V(\ell,N)} {E^{\rm bt}_V (\ell,N)}=1. 
\end{equation}
\end{theorem}

As mentioned previously, the lower bound 
\beq \label{eq: lowerbound}E_V (N,\ell) \geq \ebt_V (N,\ell) (1+o(1))\eeq
is already contained in \cite[Corollary~2.3]{LieRouYng-17}.
The subject of the present paper is the corresponding upper bound, based on a trial state argument. The issue is to approximate the density of the bathtub minimizer $\rho^{\rm bt}_V$ using states of the form \eqref{eq:fullcorr qh}.

\begin{theorem}[\textbf{Reaching the bathtub energy using uncorrelated quasi-holes}]\label{thm:upbound}\mbox{}\\
Under the stated assumptions on $V$ there exists a (sequence of) polynomial(s) $f(z)$ such that, denoting by $\rho_f$ the one-particle density of the corresponding state~\eqref{eq:fullcorr qh}, we have
\begin{equation}\label{eq:ener up bound}
e_V (N,\ell) \leq  \int_{\R^2} V (z) \rho_f (z) dz \leq \ebt_V (N,\ell) \left( 1+O(N ^{-1/4}) \right) 
\end{equation}
 in the limit $N\to \infty$.
\end{theorem}

Theorem \ref{thm:upbound} together with \eqref{eq: lowerbound} and \eqref{eq:trivial}  proves Theorem~\ref{thm:main}. A 
particular 
instance of this result was proved previously in~\cite{RouSerYng-13b}, where we considered radial potentials only, increasing or mexican-hat-like.

\medskip

We also state and prove a corollary regarding the densities of  approximate minimizers of the original problem~\eqref{eq:start ener}. Again it is convenient to use scaled variables: Define, for a given state~\eqref{eq:fullcorr} 
\begin{equation}\label{eq:scale dens intro}
 \mu_F ^{(1)} (x) = \rho_F (\sqrt{N} x )
\end{equation}
and observe that the bathtub minimizer $\rhobt_V$ is given by 
\begin{equation}\label{eq:scale bathtub}
\rhobt_V (\sqrt{N} x)= \rhobt_U (x)
\end{equation}
where $\rhobt_U$ is the minimizer of the scaled bathtub problem:
\begin{equation}\label{eq:bathtub intro U}
\ebt_U (\ell) := \inf \left\{ \intR \rho (x) U(x) dx \: | \: 0 \leq \rho \leq \frac{1}{\pi \ell},\: \intR \rho = 1  \right\} = \intR \rhobt_U (x) U (x) dx. 
\end{equation}
Note also that $E^{\rm bt}_V(\ell)=NE^{\rm bt}_U(\ell)$. 


\begin{corollary}[\textbf{Convergence of densities}]\label{cor:density}\mbox{}\\
Let $F$ be a (sequence of) correlation factors such that the associated $\Psi_F$ of the form~\eqref{eq:fullcorr} satisfy 
\begin{equation}\label{eq:almost min}
\intR V \rho_F = E_V (N,\ell) + o (N) 
\end{equation}
in the limit $N\to \infty$. Then
\begin{equation}\label{eq:conv dens}
\mu_F ^{(1)} \rightharpoonup \rhobt_U
\end{equation}
weakly as probability measures, i.e, the integrals against  any continuous bounded function converge.
\end{corollary}


Most of the paper is concerned with the proof of Theorem~\ref{thm:upbound} and we sketch here the main ideas. 
The starting point is Laughlin's {\em plasma analogy}~\cite{Laughlin-83,Laughlin-87} where the $N$-particle density of the state  \eqref{eq:fullcorr qh} is written as a Boltzmann-Gibbs factor for a classical 2D Jellium\footnote{Classical charged particles in a uniform neutralizing background of opposite charge.} with additional repulsive point charges fixed at the locations of the quasi-holes, i.e.,  at the zeros of $f$.  Following the method of~\cite{RouSerYng-13b} we shall in Section 4 rigorously justify a mean-field/zero-temperature approximation\footnote{We note that the zero-temperature approximation is only valid if the degree of the polynomial $f$ is not too large, see~\cite{RouSerYng-13a,RouSerYng-13b} for a discussion of this point.}  for this classical problem as $N\to\infty$. This yields, for a given~$f$,  a mean-field approximation of $\rho_f$ given by the unique solution of the variational equation for an electrostatic minimization problem. 
 
 For the construction of adequate trial states giving the upper bound \eqref{eq:ener up bound} we now face an {\em inverse} problem: Given the desired charge density profile $\rhobt_U$, find a distribution of repulsive point charges (quasi-holes, i.e., zeros of $f$) whose addition to the usual Jellium Hamiltonian deforms the Laughlin droplet to $\rhobt_U$.  In contrast to the direct problem the solution is not unique but, as far as the upper bound for the energy is concerned, any solution will do. In~\cite{RouSerYng-13b} we constructed solutions yielding very specific, radial, charge density profiles $\rhobt_U$. Constructing solutions yielding a general density profile $\rhobt_U$ is the main addition of the present paper. 

If we allow the distribution of the quasi-hole charges in the plasma Hamiltonian to be continuous, a simple solution to the variational equation can be given explicitly: We enclose the support $\Omega_0$ of $\rhobt_U$ in a disk $D(0,R)$ of some radius $R$ centered at the origin and fill the  complementary set $D(0,R)\setminus \Omega_0$ with a uniform distribution of quasi-holes. If the quasi-hole charge density is chosen so that the total charge density in $D(0,R)$ is constant, then Newton's Theorem ensures that the variational equation is fulfilled and its solution is $\rhobt_U$.  Further solutions to the inverse problem, also rooted in Newton's Theorem, are presented in the Appendix.

The next step is to approximate the  continuous distribution of quasi-hole charges by a discrete one. We achieve this by putting discrete charges on a lattice with suitable spacing tending to zero, but other discretizations are also possible. Requiring the polynomials in Theorem~\ref{thm:upbound} to vanish at the lattice points in $D(0,R)\setminus \Omega_0$ leads to the desired result in the limit of zero spacing.

In Section~\ref{sec:refine} we discuss possible refinements of our results to accommodate $N$-dependent potentials with variations on mesoscopic scales. 

\section{Plasma analogy and the inverse electrostatic problem}\label{sec:electrostatic}

From now on we shall mainly work with the scaled variables $x=z/\sqrt N$ and particle densities will be normalized to $1$, i.e., they are probability measures.  Like in our previous papers~\cite{LieRouYng-17,RouYng-14,RouYng-15} we shall rely on the plasma analogy mentioned above, writing the squares of many-body wave functions as Boltzmann-Gibbs factors for a one-component Coulomb gas with external charges.  The Hamiltonian defined below,  corresponding to functions of the type~\eqref{eq:fullcorr qh}, has in the scaled variables a form appropriate for a mean field approximation in the limit $N\to \infty$.

\subsection{The plasma Hamiltonian}

A polynomial $f(z)$ in \eqref{eq:fullcorr qh} can be factorized as
\begin{equation}\label{eq:qhfactor} 
f(z)=c_N\prod_{j=1}^J (z-\sqrt N a_j)^{N q_j/2}  
\end{equation}
with complex numbers $a_1,\ldots,a_J$ and integers $Nq_j/2$. The scaled $N$-particle probability density corresponding to~\eqref{eq:fullcorr qh}, 
\begin{equation}\label{eq:scaled prob}
 \mu_f ^{(N)} (x_1,\ldots,x_N) := N ^N | \Psi_f (\sqrt{N} x_1 ,\ldots, \sqrt{N} x_N)| ^2,
\end{equation}
is a probability measure on $\R^{2N}$ which can be written as
\begin{equation}\label{eq:plasma analogy}
\mu_f ^{(N)}(x_1,\ldots,x_N) = \frac{1}{\mathcal{Z}_N} \exp\left( - N \mathcal H_N(x_1,\dots, x_N) \right) 
\end{equation}
with $\mathcal{Z}_N$ a normalization factor (partition function) and 
\begin{equation}\label{eq:coul hamil}
\mathcal H_N(x_1,\dots, x_N)=\sum_{i=1}^N\left(\sum_{j=1}^J q_j\log \frac{1}{|x_i-a_j|}+|x_i|^2\right) + \frac {2\ell}N\sum_{k<l}\log \frac{1}{|x_k-x_l|}. 
\end{equation}
The auxiliary, classical hamiltonian function \eqref{eq:coul hamil} describes $N$ point charges, interacting through a 2D Coulomb potential with coupling constant ${2\ell}/N$ and with the external potential 
\beq 
\sum_{j=1}^J q_j\log \frac{1}{|x-a_j|}+|x|^2. 
\eeq
The one-particle density $\rho_f$ of~\eqref{eq:fullcorr qh} and its scaled version $\mu_f ^{(1)}$, satisfy 
\begin{equation}\label{eq:scale one part}
\rho_f (\sqrt{N} x ) = \mu_f ^{(1)} (x) = \int_{\R^{2(N-1)}} \mu_f ^{(N)} (x,x_2,\ldots,x_N) dx_2\ldots dx_N.
\end{equation}
Note that due to the scaling, the temperature in the classical Gibbs state~\eqref{eq:plasma analogy} is 
$$ T = N ^{-1},$$
and thus small in the limit $N\to \infty$. Also the coupling constant is $O(N^{-1})$. Thus a zero-temperature mean-field procedure can be expected to be adequate to determine the scaled 1-particle density $\mu^{(1)}_f(x)$ and this will be proved rigorously in Section 4. In the present section we calculate and optimize the mean-field densities associated with suitable polynomials $f$. 

\medskip

For a given $f$ as in \eqref{eq:qhfactor} the mean-field density is found by minimizing the mean field energy functional corresponding to \eqref{eq:coul hamil} for a 1-particle probability density $\rho$: 
\begin{equation} \label{eq:mffuncf}
\mathcal E_f^{\rm el}[\rho] = \int_{\R^2} \left(\Phi_{{\rm qh},f}(x)+|x|^2\right)\rho(x) dx + \ell\iint_{\R^2 \times \R ^2} \rho(x)\log\frac 1{|x-y|}\rho(y) dx dy
\end{equation}
where 
\begin{equation} \label{eq:potf}
\Phi_{{\rm qh},f}(x)=\int_{\R^2} Q_{{\rm qh},f}(y)\log\frac 1{|x-y|} dy
\end{equation}
is the potential of the quasi-hole charge density 
\begin{equation} \label{eq:chargef}
 Q_{{\rm qh},f}(x)=\sum_{j=1} ^J q_j\delta(x-a_j). 
\end{equation}
As in~\cite{RouSerYng-13a,RouSerYng-13b}, the label \lq\lq el\rq\rq\ indicates that the functional is the \lq\lq electrostatic\rq\rq\ part of the full mean-field free energy functional, 
\begin{equation}\label{eq:mffull}
\mathcal F^{\rm MF}_f[\rho]= \cEel_f[\rho]+N^{-1}\int_{\R^2} \rho(x)\log\rho(x) dx,
\end{equation}
where the last term is the entropic contribution to the free energy.

In~\cite{RouSerYng-13b} both functionals,  $\cEel_f$ and $ \mathcal F^{\rm MF}_f$,  were studied for the special case of a single quasi-hole at the origin, i.e., $J=1$, $a_1=0$ and $q_1\ll N$.  In particular it was shown that the entropic contribution can be neglected in the limit $N\to\infty$ and that  $\cEel_f$ leads to a good approximation of the quantum mechanical one-particle density in this special case. 

As discussed in Section 4 the method generalizes to  the factors~\eqref{eq:qhfactor} provided $J q_j\ll N$ and the positions $a_j$ of the quasi holes stay within a disk of radius $\ll N^{1/2}$ for all $j$. The upshot is as follows: As $N\to\infty$ the unique normalized minimizer $\rhoel_f$ of the functional~\eqref{eq:mffuncf}, satisfying
\begin{equation}\label{eq:min ener class}
\mathcal{E} ^{\rm el}_f [\rhoel_f]= \min \left\{ \mathcal{E} ^{\rm el} [\rho], \: \intR \rho = 1,\: \rho \geq 0\right\} =: \Eel_f, 
\end{equation}
is a good approximation of  $\mu^{(1)}_f(x)=\rho_f(\sqrt N x)$. Moreover, the  density $\rhoel_f$ takes the constant value $(\ell\pi)^{-1}$ where it is nonzero. The problem we address in the present section is to suitably arrange the quasi-holes charge distribution~\eqref{eq:chargef} so that $\rhoel_f$ approximates the (scaled) bathtub minimizer $\rho^{\rm bt}_U$.

\subsection{Approximating the bathtub minimizer}
A  minimizer for the bathtub energy \eqref{eq:scale bathtub} it is explicitly given as follows~\cite[Theorem~1.14]{LieLos-01}: For a smooth potential  $U$, bounded below,  the sublevel sets
\begin{equation}\label{eq:sub level}
\{x:\, U(x)\leq e\} 
\end{equation}
increase  from the empty set to the whole of $\R^2$ as $e$ increases from below the minimum value of $U$ to $\infty$. Let $e_0$ be the smallest energy value such that the area $|\{x:\, U(x)\leq e_0\}|$  is $\geq \pi\ell$. Then there is a (possibly not unique)  subset $\Omega_0$ of this sublevel set with $|\Omega_0|=\pi\ell$ as well as $\{x:\, U(x)< e_0\}\subset\Omega_0$. Note that $\Omega_0$ need not be connected. The corresponding density  
\begin{equation}\label{eq:bathtubrho}
\rho_0(x)=\begin{cases} 
(\pi\ell)^{-1} &\hbox{ if }x\in\Omega_0
\\ 0 &\hbox{ otherwise }  \end{cases} 
\end{equation}
minimizes the bathtub energy~\eqref{eq:bathtub intro U}. By our assumption that $U$ has no flat pieces, $\Omega_0$ is in fact unique and $\rho_0=\rho^{\rm bt}_U$.

We approximate $\rho_0$ in the metric defined by the Coulomb kernel: If $\sigma$ is any finite, signed measure with $\int |\log|x| \sigma(x)| dx<\infty$ we define
\begin{equation}\label{eq:Dsigma}
D(\sigma,\sigma):=\frac 12\iint_{\R^2 \times \R ^2} \sigma(x) \log\frac 1{|x-y|}\sigma(y)\,dx\,dy. 
\end{equation}
If $\int\sigma(x)dx=0$, in particular if  $\sigma=\rho_1-\rho_2$ with two probability distributions $\rho_1, \rho_2$, then $D(\sigma,\sigma)\geq 0$ and 
\begin{equation} \label{eq:metric} 
d(\rho_1,\rho_2):=D(\rho_1-\rho_2,\rho_1-\rho_2)^{1/2} 
\end{equation}
is a metric on the set of probability measures. If $\chi$ is a differentiable function, then by Fourier transform and Cauchy-Schwarz inequality 
\begin{equation}\label{eq:CS}
 \left|\int (\rho_1(x)-\rho_2(x))\chi(x) dx\right|\leq C\, d(\rho_1,\rho_2)\,\Vert\nabla \chi\Vert_{L^2}.   
\end{equation}
The main result of this section is the following:

\begin{proposition}[\textbf{Inverse electrostatic problem}]\label{pro:inverse el}\mbox{}\\
With the previous assumptions and notation, there exists a (sequence of) polynomial(s) $f_\delta$ indexed by a small $N$-dependent parameter $\delta >0$ such that, denoting $\rho_\delta$ the corresponding electrostatic minimizer of~\eqref{eq:mffuncf}, we have
\begin{equation}\label{eq:inverse el}
D(\rho_0-\rho_\delta,\rho_0-\rho_\delta) \leq C N ^{-1/2}
\end{equation} 
where $C$ is a constant depending only on $\ell$ and $U$.
\end{proposition}

The rest of this section is concerned with the proof of this result. As previously mentioned, if the quasi-holes charge density is allowed to be continuous, one can achieve identity, $\rhoel = \rho_0$. We discuss this first, and then turn to approximating the exact solution using a discrete set of point charges.

\subsection{Smeared quasi-hole charges}\label{sec:inv prob}

Let us generalize~\eqref{eq:mffuncf} and consider
\begin{equation}\label{eq:mffunc} 
 \cEel[\rho]=\int_{\R^2} \left(\Phi_{\rm qh}(x)+|x|^2\right)\rho(x) dx + 
2\ell D(\rho,\rho)
\end{equation}
where 
\begin{equation}\label{eq:gen pot}
 \Phi_{\rm qh}(x)=\int Q_{\rm qh}(y)\log\frac 1{|x-y|} dy 
\end{equation}
is the potential of an arbitrary  positive measure $Q_{\rm qh}(x)$ of finite mass. In particular, $Q_{\rm qh}$ can be a measurable, positive function,  but the discrete measure~\eqref{eq:chargef}, or a mixture of discrete and continuous parts as discussed in Appendix~\ref{sec:app}, are also included. The subsidiary conditions for the minimization problem are
\begin{equation}\label{eq:constraints}
 \rho\geq 0,\qquad \int_{\R^2} \rho=1.
\end{equation}
We use the notation 
\begin{equation}\label{eq:def pot}
\Phi_\rho=\rho*\log\frac{1}{|\cdot|} 
\end{equation}
for the potential generated by a charge density $\rho$.

\begin{lemma}[\textbf{Inverse problem with smeared charges}]\label{lem:smeared}\mbox{}\\
Let $D(0,R)$ be a disk with center at the origin containing $\Omega_0$, the support of the bathtub minimizer $\rho_0$. Define
\begin{equation}\label{eq:cont charge}
Q_{\rm qh}(x) = Q_0(x):=\begin{cases} 
2/\pi &\hbox{ if }x\in D(0,R)\setminus\Omega_0 
\\ 0 &\hbox{otherwise }  \end{cases}.
 \end{equation}
Then, the corresponding unique minimizer of~\eqref{eq:mffunc} is equal to $\rho_0$.
\end{lemma}

\begin{proof}
By standard arguments, one sees that the functional~\eqref{eq:mffunc} with a general $Q_{\rm qh}$ is bounded from below, strictly convex and has a unique minimizer, $\rho$, determined by $Q_{\rm qh}$. The variational equation for the minimizer reads
\begin{equation} \label{eq:var}
\Phi_{\rm qh}(x)+|x|^2 +
2\ell\Phi_\rho(x)
=\begin{cases} 
C^{\rm el} &\hbox{ if } \rho(x)>0
\\ \geq C^{\rm el} &\hbox{ if } \rho(x)=0 \end{cases} 
\end{equation}
with
\begin{equation}\label{eq:varconstant}
 C ^{\rm el}=\cEel[\rho]+
 2\ell D(\rho,\rho).
\end{equation}
Using the strict convexity of the functional (see~\cite[Theorem~1.2]{ChaGozZit-13} or~\cite[Theorem~II.10]{LieSim-77b}) one can show that the variational equation determines the minimizer uniquely. In particular, if $Q_{\rm qh}=Q_0$ and a density $\rho$ satisfies~\eqref{eq:var} with some constant $C$, then $\rho = \rho_0$ and $C = C ^{\rm el}$ is given by~\eqref{eq:varconstant}. 

Applying the Laplacian to \eqref{eq:var} gives
\begin{equation}\label{eq:delta var}
\rho(x)=\frac 1{\pi\ell}-\frac 1{2\ell} Q_{\rm qh}(x)\,\quad\text{ if }\rho(x)\neq 0. 
\end{equation}
Thus, if $Q_{\rm qh}(x)\geq 2/\pi$, then $\rho$ and $Q_{\rm qh}$ have disjoint supports (because $\rho\geq 0$), and\footnote{The rigorous derivation of~\eqref{eq:value} needs some care because the terms in~\eqref{eq:var} are not twice continuously differentiable, but one can argue as in~\cite[Equations (3.42)-(3.43)]{LieRouYng-17}, based on arguments from~\cite{FraLie-16}.}
\begin{equation}
\label{eq:value} \rho(x)=(\pi\ell)^{-1} 
\end{equation}
a.e.\ where $\rho\neq 0$. 

Now, with the definition~\eqref{eq:cont charge} we have
\begin{equation}
 Q_0(x)+2\ell\rho_0(x)=
 \begin{cases} 2/\pi &\hbox{ if }x\in D(0,R)
 \\ 0 &\hbox{otherwise }
 \end{cases} 
\end{equation}
and one can calculate the potential associated to the rotationally symmetric $Q_0 +2\ell\rho_0$ using Newton's theorem~\cite[Theorem~9.7]{LieLos-01}. In particular, it is constant on $D(0,R)$ which includes the support of $\rho_0$. This shows that~\eqref{eq:var} holds for $\rho=\rho_0$ with
$C^{\rm el}=C_R=R^2-2R^2\log R$ and we deduce that indeed $\rho_0$ is the unique minimizer of~\eqref{eq:mffunc}.
\end{proof}

Some alternatives to the solution \eqref{eq:cont charge} are discussed in Appendix~\ref{sec:app}.

\subsection{Discrete quasi-hole charges}

We now complete the 

\begin{proof}[Proof of Proposition~\ref{pro:inverse el}]
The charge density \eqref{eq:chargef}, corresponding to the  quasi-hole factor  \eqref{eq:fullcorr qh}, describes  discrete point charges of magnitude $q_j$ at positions $a_j$. For a given discretization scale $\delta>0$ we would like to approximate the continuous  distribution $Q_0$ with a discrete one, denoted by $Q_\delta$, such that the corresponding minimizer $\rhoel_f$ of  \eqref{eq:mffunc} with $Q_{\rm qh}=Q_\delta$ approximates the bathtub density $\rho_0$ with controllable errors as $\delta\to 0$. The simplest way is to take all $q_j$ equal to the smallest possible value it can have such that $Nq_j/2$ is a positive integer:
\beq 
q_j=2/N\quad\text{for all $j$},
\eeq
and distribute the points $a_j$ on a grid in the complementary set $\Omega_0'=D(0,R)\setminus\Omega_0$ with lattice constant $\delta\to 0$. The discrete charge density is now a sum of delta-functions,\begin{equation}\label{eq:disk charge}
Q_\delta(x) = \frac{2}{N}\sum_{j=1}^M \delta(x-a_j), 
\end{equation}
corresponding to the quasi-hole factor
\begin{equation}\label{eq: fqh delta}
f_\delta(z)=c_N\prod_{j=1}^M(z-\sqrt N a_j). 
\end{equation}
Thus, in the notation \eqref{eq:chargef}, $Q_\delta$ is shorthand for $Q_{{\rm qh}, f_\delta}$. We denote the corresponding potential \eqref{eq:gen pot} by $\Phi_\delta$. Likewise $\Phi_0$ stands for the potential generated by $Q_0$.

In order that $\Omega_0'$ is approximately covered by the grid and the average charge density is the same as for $Q_0$, i.e., $2/\pi$, we must have
\beq M\delta^2\simeq |\Omega_0'|\quad\text{and}\quad \frac {M\cdot(2/N)}{M\cdot \delta^2}=\frac 2\pi\eeq
which means

\begin{equation}\label{eq:choose grid}
M\simeq \frac N\pi |\Omega_0'|\quad\text{and}\quad\delta =\delta_N= \sqrt{\frac{\pi}{N}}. 
\end{equation}

We take the $a_j$ to be the mid-points of the squares fully included in $\Omega_0'$ (that is, not intersecting the boundary) labeled by $j=1 \ldots M$. In this way we have
\begin{equation} \label{eq:Mq} 
 M\delta^2=|\Omega_0'| + O(N^{-1/2})
\end{equation}
because the area of the part of $\Omega_0'$ (which is a \emph{fixed}, regular set) that does not get covered in this procedure is clearly bounded above by the length of the boundary $\partial\Omega_0'$ of $\Omega_0'$ times the side length $\delta\sim N^{-1/2}$ of a square of the grid.

If $g$ is a differentiable function then
 \begin{equation}\label{eq:riemann}
 \left|\intR (Q_0(x)-Q_\delta(x)) g(x) dx\right| \leq  C N^{-1/2} \,|\Omega_0'|\,  \sup_{\Omega'_0}|\nabla g| + C N ^{-1/2} \,|\partial\Omega_0'|\,\sup_{\Omega_0'} |g|. 
 \end{equation}
The first error term comes from approximating $g$ by a constant in each square of the grid, and the second one from the part of $\Omega_0'$ not covered. 
 
The minimizers $\rho_0$ and $\rho_\delta$ corresponding  respectively to $Q_0$ and $Q_\delta$  satisfy (cf.\ \eqref{eq:value})
\begin{equation}\label{eq:densbounds}
\Vert \rho_0\Vert_{L^1}=\Vert \rho_\delta\Vert_{L^1}=1 \mbox{ and } \Vert \rho_0\Vert_{L^\infty}=\Vert \rho_\delta\Vert_{L^\infty}=\frac1{\pi\ell}
\end{equation}
which implies, for the associated potentials,
\begin{equation}\label{eq:potbound} 
\Vert\nabla \Phi_{\rho_0} \Vert_{L^\infty}\leq\frac{2\pi+1}{(\pi\ell)^{1/2}} \quad \mbox{ and } \quad \Vert\nabla \Phi_{\rho_\delta} \Vert_{L^\infty} \leq\frac{2\pi+1}{(\pi\ell)^{1/2}}.
\end{equation}
Indeed, if $g=h*\log\frac 1{|\,\cdot\,|}$, then for all $r>0$
\begin{align}\label{eq:bound pot}
|\nabla g(x)|&\leq \intR \frac {|h(y)|}{|x-y|}dy = \int_{|x-y|\leq r} \frac{|h(y)|}{|x-y|}dy+\int_{|x-y|\geq r} \frac {|h(y)|}{|x-y|}dy\nonumber\\
&\leq 2\pi r\Vert h\Vert_\infty+\frac 1r\Vert h\Vert_1 
\end{align}
and optimization over $r$, using  \eqref{eq:densbounds} for $h=\rho_0$ and $h=\rho_\delta$ respectively, gives~\eqref{eq:potbound}. By a similar argument we obtain that 
\begin{equation}\label{eq:potbound bis} 
\Vert \Phi_{\rho_0} \Vert_{L^\infty}\leq C \quad \mbox{ and } \quad \Vert \Phi_{\rho_\delta} \Vert_{L^\infty} \leq C 
\end{equation}
where the constant $C$ depends only on $\ell$ and $R$.

Using $\rho_0$ as a trial density for the energy functional $\cEel_{\delta}$ given by \eqref{eq:mffunc} with $\Phi_{\rm qh}=\Phi_\delta$ we get
\begin{equation}\label{eq:approx1}
\cEel_\delta[\rho_\delta]\leq \cEel_\delta[\rho_0] = \cEel_0[\rho_0]+ \intR \left(\Phi_\delta- \Phi_0\right) \rho_0. 
\end{equation}
But 
\beq \intR \left(\Phi_\delta- \Phi_0\right) \rho_0  = \intR \left( Q_\delta - Q_0 \right) \Phi_{\rho_0}\eeq
so that, employing~\eqref{eq:riemann},~\eqref{eq:potbound} and~\eqref{eq:potbound bis} we obtain 
\begin{equation}\label{eq:approx1 bis}
\cEel_\delta[\rho_\delta]\leq \cEel_0[\rho_0] + C N^{-1/2}. 
\end{equation}
Similarly, using $\rho_\delta$ as a trial density for the functional $\cEel_0$  with $\Phi_{\rm qh}=\Phi_0$,
\begin{equation}\label{eq:approx2}
\cEel_0[\rho_0]\leq \cEel_0[\rho_\delta] \leq \cEel_\delta [\rho_\delta] + CN ^{-1/2} 
\end{equation}
so that the energies coincide in the limit. On the other hand, using the variational equation~\eqref{eq:var}, we have the stability result \begin{equation}\label{eq:stability}
\cEel_\delta[\rho_0] \geq \cEel_\delta [\rho_\delta] + 2 \ell D\left( \rho_0- \rho_\delta, \rho_0- \rho_\delta \right) 
\end{equation}
which, upon combining with~\eqref{eq:approx1 bis} and~\eqref{eq:approx2} gives the desired bound 
\eqref{eq:inverse el} for the density.
The proof of~\eqref{eq:stability} follows~\cite[Section~3.2]{RouSerYng-13b}: 
Define 
\beq W_\delta(x):=\Phi_\delta(x)+|x|^2\eeq
and write the variational equation~\eqref{eq:var} for $\rho_\delta$ as
\begin{equation}\label{eq:var2}
W_\delta(x)+2\ell\, \Phi_{\rho_\delta}(x)=
\begin{cases} 
C^{\rm el} &\hbox{ if } \rho_\delta(x)>0\\ 
\geq C^{\rm el} &\hbox{ if } \rho_\delta(x)=0 
\end{cases}
\end{equation}
Consider now a variation $\nu$ of $\rho_\delta$ with
\begin{equation}\label{eq:nucond}
\rho_\delta+\nu\geq 0,\quad \intR \nu=0. 
\end{equation}
Using~\eqref{eq:var2} and~\eqref{eq:nucond}, noting that $\rho_\delta+\nu\geq 0$ implies that $\nu\geq 0$ where $\rho_\delta=0$, we get 
\begin{align}\label{Dbound}
\cEel_\delta[\rho_\delta+\nu]&=\cEel_\delta [\rho_\delta] + \intR \left(W_\delta +2\ell\, \Phi_{\rho_\delta}\right)\nu +2\ell D(\nu,\nu)\nonumber\\
&\geq \cEel_\delta [\rho_\delta] + C^{\rm el} \intR \nu+ 2\ell D(\nu,\nu) = \cEel_\delta [\rho_\delta] + 2\ell D(\nu,\nu).
\end{align}
which is \eqref{eq:stability} if $\nu=\rho_0-\rho_\delta$.

\end{proof}

\section{Completion of the proofs}\label{sec:conc proof}

In Proposition \ref{pro:inverse el} we considered the approximation of the bathtub density $\rho_0=\rho^{\rm bt}_U$ by the mean-field density $\rho_\delta=\rho^{\rm el}_{f}$ with $f=f_\delta$ given by \eqref{eq: fqh delta}. 
We now supply the missing piece in the proof of our main result, Theorem \ref{thm:upbound}, namely the rigorous justification of the approximation of the true quantum mechanical one-particle density $\mu^{(1)}_f$ (cf.\ Eq.\ \eqref{eq:scale one part})  by the mean field density $\rho^{\rm el}_{f}$. This part of our analysis follows closely the methods of ~\cite{RouSerYng-13b}. 

\subsection{The mean-field approximation}

The link between the quantum mechanical trial states we use in the proof of Theorem~\ref{thm:upbound} and the mean-field plasma problem discussed in Section~\ref{sec:electrostatic} is as follows: 

\begin{proposition}[\textbf{Mean-field approximation for the QM density}]\label{pro:MF plasma}\mbox{}\\
Let $f = f_\delta$ be defined as in~\eqref{eq: fqh delta} with the choices discussed in Section~\ref{sec:inv prob}. With the previous notation, we have, for any test function $\chi$,
\begin{equation}\label{eq:MF final}
\left|\intR \left(\mu^{(1)}_{f}(x)-\rhoel_{f}(x)\right)\chi(x) dx\right| \leq C \left( \frac{\log N}{N} \right) ^{1/2} \Vert \nabla \chi\Vert_{L^2} + C N^{-1/2}\Vert \nabla \chi\Vert_{L ^\infty}. 
\end{equation}
Moreover, we have the pointwise decay estimate 
\begin{equation}\label{eq:mu-decay}
0 \leq \mu^{(1)}_{f}(x)\leq Ce^{-N C(|x|^2-\log N)}. 
\end{equation}
\end{proposition}

\begin{proof}[Proof of Proposition~\ref{pro:MF plasma}]
In an intermediate step we rely on the mean-field \emph{free}-energy functional~\eqref{eq:mffull}, denoted by $\cEMF_f$ where $f$ is the $f_\delta$ constructed in Proposition~\ref{pro:inverse el}.  We denote by $\rhoMF_f$ the (unique) minimizer of $\cEMF_f$ amongst probability measures, and  the associated minimal free-energy by $\EMF_f$. Recall also that $\rho^{\rm el}_f=\rho_\delta$ is, by definition, the minimizer of \eqref{eq:mffuncf}, i.e., the electrostatic part of the mean field functional.

\medskip 

\noindent\textbf{Step 1.} The proof of~\cite[Theorem~3.2]{RouSerYng-13b} for the special case $f(z)=z^m$, $m\ll N^2$, carries over {\em mutatis mutandis} to the present situation, and yields 
\begin{equation}\label{eq:MFapprox}
\left|\intR \left(\mu^{(1)}_{f} (x)-\rhoMF_f (x)\right) \chi(x) dx\right | \leq C\left(\frac {\log N} N\right)^{1/2}\Vert \nabla \chi\Vert_{L^2} + C N^{-1/2} \Vert \nabla \chi\Vert_{L ^\infty}. 
\end{equation}
Thus, we only need to estimate the difference between the free-energy minimizer $\rhoMF_f$ and the 
\lq\lq electrostatic\rq\rq\ energy minimizer $\rho_f^{\rm el}$. 

\medskip 

\noindent\textbf{Step 2.} We claim that
\begin{equation}\label{eq:Dconv3}
D(\rhoMF_f - \rho^{\rm el}_f, \rhoMF_f - \rho^{\rm el}_f)\leq C N^{-1}.
\end{equation}
For the proof, we define  
\begin{equation}\label{eq:etildefunc} 
\cEelt_{f}[\rho] = \cEel_{f} [\rho]- N^{-1}\intR |x|^2\rho(x) dx
\end{equation}
with minimizer $\rhoelt_f$ and minimal value $\Eelt_f$. We have the string of inequalities 
\begin{align}\label{eq:string}
\Eel_f - N^{-1}\log(\pi\ell) &= \Eel_{f} + N^{-1} \intR \rhoel_{f}\log  \rhoel_{f} \nonumber \\ 
&\geq \EMF_f = \cEelt_f [\rhoMF_f] + N ^{-1} \intR \rhoMF_f \log \frac{\rhoMF_f}{\pi ^{-1} e^{-|x| ^2 }} - N ^{-1} \log \pi \nonumber\\
&\geq \Eelt_f + D \left(\rhoMF_f - \rhoelt_f, \rhoMF_f - \rhoelt_f \right) - N ^{-1} \log \pi \nonumber \\
&= \cEel_f [\rhoelt_f] - N^{-1} \intR |x|^2\rhoelt_f (x) dx \nonumber\\
&+ D \left(\rhoMF_f - \rhoelt_f, \rhoMF_f - \rhoelt_f \right) - N ^{-1} \log \pi \nonumber\\
&\geq \Eel_f + D \left(\rhoelt_f - \rhoel_f, \rhoelt_f - \rhoel_f \right) + D \left(\rhoMF_f - \rhoelt_f, \rhoMF_f - \rhoelt_f \right) \nonumber\\
&- N ^{-1} \log \pi - N^{-1} \intR |x|^2\rhoelt_f (x) dx.
\end{align}
In the first line we just use that $\rhoelt_f$ is constant on its support, and the first inequality is the variational principle. The last two inequalities follow from positivity of the relative entropy and~\eqref{eq:stability} applied to the functional~\eqref{eq:etildefunc} and then to $\cEel$. All in all, using the triangle inequality for the Coulomb metric \eqref{eq:metric}, we are left with estimating the last term of~\eqref{eq:string}, which is done by the following virial-type argument:

Consider the scaled densities 
\begin{equation}\label{eq:virial}
\rhoelt_{f,t} (x):=t^2\rhoelt_f(tx)
\end{equation}
with $t\geq 0$ and use that, because $\rhoelt_{f,1} = \rhoelt_f$ is a minimizer,
\begin{equation}\label{eq:virial 2}
 \frac{d}{dt} \, \Eelt_f \left[\rhoelt_{f,t}\right] _{|{t=1}}\geq 0.
\end{equation}
A calculation allows to deduce
\begin{equation}
 -2(1-N^{-1})\intR |x|^2 \rhoelt_f(x)dx + \iint_{\R^2 \times \R ^2} Q_{\delta}(y) \frac{|x|^2-\half {\rm Re}\, \bar x y}{|x-y|^2}\,\rhoelt_f(x)\, dx \,dy + 2\ell \geq 0
\end{equation}
which, because $\tilde \rho\in L^1\cap L^\infty$, implies (by considerations similar to~\eqref{eq:bound pot}) 
\beq\intR |x|^2\rhoelt_f(x)dx \leq C_1\intR Q_{{\rm qh},f} + C_2 \leq C R^2,\eeq
where $R$ is the radius of the enclosing disk in Lemma~\ref{lem:smeared}. This concludes the proof of~\eqref{eq:Dconv3}, for the disc is fixed. 

\medskip 

\noindent\textbf{Step 3.} To complete the proof of~\eqref{eq:MF final} we simply write 
\begin{multline}
 \intR \left(\mu^{(1)}_{f} (x)-\rhoel_f (x)\right) \chi(x) dx = \intR \left(\mu^{(1)}_{f} (x)-\rhoMF_f (x)\right) \chi(x) dx \\
 + \intR \left(\rhoMF_f (x)-\rhoel_f (x)\right) \chi(x) dx. 
\end{multline}
To estimate the first term we use~\eqref{eq:MFapprox}. For the second one we combine~\eqref{eq:Dconv3} and~\eqref{eq:CS}.

\medskip 

\noindent\textbf{Step 4.} Finally, we turn to the decay estimate~\eqref{eq:mu-decay}. As in~\cite{RouSerYng-13b}, it is sufficient to provide a decay estimate on the free-energy minimizer $\rhoMF_f$. The variational equation for the latter reads
\begin{equation}\label{eq:var T}
\Phi_{\rm qh}(x)+|x|^2+2\ell\,\Phi_{\rhoMF_f} (x) + N^{-1} \log\rhoMF_f (x) =  C^{\rm MF}
\end{equation}
with 
$$ C^{\rm MF} = \EMF_f + \ell D (\rhoMF_f,\rhoMF_f). $$
Thus,
\begin{equation}\label{MFformula}
\rhoMF_f (x)=\exp\left[-N\left(|x|^2+\Phi_{\rm qh}(x)+2\ell\,\Phi_{\rhoMF_f}(x)-C^{\rm MF} \right)\right].
\end{equation}
It is easy to see that $C^{\rm MF}$ is bounded independently of $N$. On the other hand, outside of $D(0,R+1)$ the potential $\Phi_{\rm qh}$ is bounded as 
$$ \left| \Phi_{\rm qh} (x) \right| \leq C \left( 1 + |\log|x|| \right). $$
Moreover, $\rhoMF_f$ is integrable\footnote{This is part of the definition of the variational set for the free energy functional, ensuring that $D(\rho,\rho)$ is well defined.} against the measure $(1 + |\log|x||)dx$ and
$$\Vert\rhoMF_f \Vert_{L ^\infty} \leq (\pi\ell)^{-1},$$ 
see~\cite[Equation~(3.19)]{RouSerYng-13b}. 
Hence $\Phi_{\rho^{\rm MF}}$ behaves as ${\rm (const.)}\log|x|$ for large $|x|$, see~\cite[Lemma~3.5]{LieRouYng-17}. Thus the term $|x|^2$ dominates the exponent in \eqref{MFformula} for large $|x|$ and we obtain
\begin{equation}\label{eq:MFdecay}
0\leq \rho^{\rm MF}(x)\leq Ce^{-NC|x|^2},
\end{equation}
This decay estimate carries over to the one-particle probability density $\mu^{(1)}_{f}$ in exactly the same way as in the proof of~\cite[Equation~(3.16)]{RouSerYng-13b}, see the end of Section 3.3 in that reference. 
\end{proof}

\subsection{Proofs of the main theorems}

We can now finish the

\begin{proof}[Proof of Theorem~\ref{thm:upbound}]
With our previous assumptions and notation, taking $f=f_\delta$ as defined in Subsection~\ref{sec:inv prob}, we have 
\begin{align}\label{eq:final split}
 \intR V (z) \rho_f (z) dz &= N \intR U(x) \mu_f ^{(1)} (x) dx  = N \intR \chiin (x) U(x) \rho_0 (x) dx \nonumber \\
 &+ N \intR \chiin (x) U(x) \left( \rho_0 (x) - \mu_f ^{(1)} (x) \right)dx + N \intR \chiout (x) U(x) \mu_f ^{(1)} (x) dx
\end{align}
where $\chiin$ and $\chiout$ are a smooth partition of unity with $\chiin$ supported in the disk $D(0, 2 \log N)$ and $\chiout$ zero in $D(0,\log N)$. Obviously we can impose 
\begin{equation}\label{eq:bound chi}
 \norm{\nabla \chiin }_{L^{\infty}} + \norm{\nabla \chiout }_{L^{\infty}} \leq C \log N. 
\end{equation}
Since $U$ is fixed and increases at infinity, the minimizer $\rho_0$ of the associated bathtub problem has compact support and thus the first term in the right-hand side of~\eqref{eq:final split} is, for large enough $N$, equal to $\ebt_V (N,\ell)$, which is proportional to $N$ by scaling.

We are thus left with estimating the two error terms. Using~\eqref{eq:growth pot},~\eqref{eq:CS},~\eqref{eq:inverse el} and \eqref{eq:MF final}, we obtain 
\begin{equation}\label{eq:4.15}
\left| \intR \chiin (x) U(x) \left( \rho_0 (x) - \mu_f ^{(1)} (x) \right)dx \right| \leq C N^{-1/4}. 
\end{equation}
On the other hand, using~\eqref{eq:mu-decay} and~\eqref{eq:growth pot} again, 
\begin{equation}
 \left| \intR \chiout (x) U(x) \mu_f ^{(1)} (x) dx\right| \leq \int_{|x|\geq \log N} \exp \left(- CN (|x| ^2 - \log N) \right) |x| ^s
\end{equation}
is clearly exponentially small in the limit $N\to \infty$. The proof is complete.
\end{proof}

Finally, let us give the 
\begin{proof}[Proof of Corollary~\ref{cor:density}]
From~\eqref{eq:almost min} and the fact that $U$ grows at infinity, it follows that the sequence 
\beq \mu_F ^{(1)} = \rho_F (\sqrt{N}\,)\eeq
is tight (precompact in the topology of weak convergence), for otherwise the energy would for large $N$ become much larger than the bathtub energy in  contradiction to Theorem~\ref{thm:upbound}. Thus $\mu_F ^{(1)}$ converges as a probability measure, along a subsequence, to a limiting $\mu_\infty$. We claim that 
\beq \mu_\infty = \rhobt_U,\eeq
which is the desired result (by uniqueness of the limit, the whole sequence then converges).

To prove the assertion we use a Feynman-Hellmann-type argument. Let $\chi$ be a smooth compactly supported test function and $\eps>0$ a small, fixed for now, real number. By assumption we have 
\beq 
\intR \mu_F ^{(1)} ( x ) \left( U (x) + \eps \chi (x) \right) dx = N ^{-1} E_V (N,\ell) + \eps \intR \mu_F ^{(1)} ( x ) \chi (x) dx + o_N (1)
\eeq
and thus Theorem~\ref{thm:main} yields 
\beq 
\intR \mu_F ^{(1)} ( x ) \left( U (x) + \eps \chi (x) \right) dx = \ebt_U (\ell) + \eps \intR \mu_F ^{(1)} ( x ) \chi (x) dx + o_N (1)
\eeq
But, using the energy lower bound from~\cite[Corollary~2.3]{LieRouYng-17}, we obtain  
\begin{align}
 \intR \mu_F ^{(1)} ( x ) \left( U (x) + \eps \chi (x) \right) dx &\geq \ebt_{U + \eps \chi} (\ell)+ o_N (1) \nonumber \\
 &\geq \ebt_U (\ell) + \eps \intR \chi \rhobt_{U+\eps \chi} + o_N (1).
\end{align}
Combining the previous estimates, letting $N\to \infty$ (along the subsequence previously identified) and then dividing by $\eps>0$ we obtain 
\beq
\intR \chi \mu_\infty \geq \intR \chi \rhobt_{U+\eps \chi}.
\eeq
Letting then $\eps \to 0$ we deduce (using the explicit formula for bathtub minimizers~\cite[Theorem~1.14]{LieLos-01})  
\beq 
\intR \chi \mu_\infty  \geq \intR \chi \rhobt_{U}.
\eeq
Repeating the previous steps with now $\eps < 0$ gives the reversed inequality and concludes the proof.
\end{proof}

\section{Refinements}\label{sec:refine}

The results in the previous sections were proved under the assumption that the scaled potential $U$ is fixed, i.e., independent of $N$. This means in particular that our analysis applies to potentials $V$ in~\eqref{eq:mag hamil} that vary locally over an arbitrary small, but {\em fixed},  fraction of the macroscopic sample scale $O(\sqrt N)$. On the other hand, to capture fully the effects of disorder caused by small, random impurities, it is desirable to consider also variations of $V$ on mesoscopic  scales $O(N^\alpha)$ with $0<\alpha<{1/2}$ ($\alpha = 0$ corresponds to the mean inter-particles distance). 

In terms of the scaled potential $U(x)=V(\sqrt N x)$ this amounts to allowing $N$-dependent $U_N$ such that $|\nabla U_N|$ is only required to be bounded locally by $N^{1/2-\alpha}$. It is still natural to require global confinement of the system on the scale $\propto \sqrt N$, and these two ideas can be formalized by writing
\begin{equation} 
U_N(x)=U_{{\rm dis},N}(x)+U_{\rm trap}(x)
\end{equation}
where the $N$-independent $U_{\rm trap}$ satisfies Assumption \ref{asum:pot}, in particular growth at infinity, while $U_{{\rm dis},N}$ is differentiable with bounded support, uniformly bounded in $N$ and satisfies
\beq \label{eq:reg cond}\Vert \nabla U_{{\rm dis},N}\Vert_\infty\leq C N^{1/2-\alpha}.\eeq
With such an assumption the explicit error bound in \eqref{eq:ener up bound} will not hold in general, of course,  but  \eqref{eq:MFapprox} with $\chi=U_{{\rm dis},N}$ would still tend to zero. In order to show that \eqref{eq:4.15} also tends to zero, which is required for Theorem~\ref{thm:upbound} even without an explicit error estimate, one has to take a closer look at the bound \eqref{eq:inverse el} on the difference between the bathtub minimizer $\rho_0=\rho^{\rm bt}_U$ and the mean-field minimizer $\rho_\delta=\rho^{\rm el}_{f_\delta}$.

With an $N$-dependent $U$ the support $\Omega_{0,N}$ of $\rho_0$ depends on $N$. Its area is still $N$-independent, $|\Omega_{0,N}|=\pi\ell$, and $\Omega_{0,N}$ stays within a fixed, $N$-independent disk $D(0,R)$ by our assumptions on $U_N$. Thus also the area of $\Omega_{0,N}'=D(R,0)\setminus \Omega_{0,N}$ does not depend on~$N$.

A problem might occur, however, if we define the discrete approximation $Q_\delta$ of the continuous distribution $Q_0$ as we did before, i.e, by tiling 
$\Omega_{0,N}'$ with squares of side length $\delta=\delta_N=(\pi/N)^{1/2}$. We recall that this  length was the smallest value of $\delta$ compatible with the fact that every exponent in \eqref{eq:qhfactor} has to be a positive integer. The estimate \eqref{eq:riemann} which is the basis for Proposition \ref{pro:inverse el}, has two error terms. The first one is proportional to $|\Omega_{0,N}'|$ and thus $N$-independent. The second one includes an estimate of the area of that part of $\Omega_{0,N}'$ which is not covered by the tiling with squares. This area is estimated from above by $\delta_N\cdot |\partial \Omega_{0,N}'|$. Now, while $|\Omega_{0,N}'|$ is independent of $N$, the same need not hold for the length of the boundary, $|\partial \Omega_{0,N}'|$. Since $R$ and hence the outer boundary $\partial D(R,0)$ of $\Omega_{0,N}'$ is fixed, this concerns only the inner boundary, which is $\partial \Omega_0$. Provided $|\partial \Omega_{0,N}|\ll N^{1/2}$ however, the right-hand side of \eqref{eq:inverse el} is at least $o(1)$ and this is sufficient for a proof of Theorem \ref{thm:upbound} with the explicit error term replaced by $o(1)$. 

The bound 
\begin{equation}\label{eq:boundary length}
|\partial \Omega_{0,N}|\ll N^{1/2} 
\end{equation}
is a mild regularity condition on the sequence of potentials $U_N$, certainly compatible with~\eqref{eq:reg cond} although it does not follow from it. Adding it to our assumptions on $U_N$ we have the following corollary of the considerations in Sections~\ref{sec:electrostatic} and~\ref{sec:conc proof}:

\begin{corollary}[{\bf Generalization of Theorem \ref{thm:upbound}}]
\label{thm:upbound gen}\mbox{}\\
Under the assumptions above on $U_N$ and with $V(z)=U_N(z/\sqrt N)$ there exists a (sequence of) polynomial(s) $f(z)$ such that, denoting $\rho_f$ the one-particle density of the corresponding state~\eqref{eq:fullcorr qh} we have
\begin{equation}\label{eq:ener up bound gen}
e_V (N,\ell) \leq  \int_{\R^2} V (z) \rho_f (z) dz \leq \ebt_V (N,\ell) \left( 1+o(1)) \right). 
\end{equation}
in the limit $N\to \infty$.
\end{corollary}

\section{Conclusions} We have derived an optimal upper bound in the large $N$ limit for the potential energy of states in the Lowest Landau Level exhibiting the same correlations as the Laughlin state. This bound is obtained using trial states deforming the Laughlin state by means of uncorrelated quasi-particles. It matches exactly a previously derived lower bound and provides a mathematical proof of Laughlin's insight that uncorrelated quasi-holes suffice for describing the response of his function to external perturbations.
 
 The method relies on rigorous mean-field theory for the Coulomb Hamiltonian in Laughlin's plasma analogy.  The construction of adequate trial states is achieved by  solving an inverse electrostatic problem.

\appendix

\section{Alternative solutions}\label{sec:app}

The solution \eqref{eq:cont charge}
is not the only quasi-hole distribution leading  to the bathtub solution \eqref{eq:bathtubrho}
 as an exact minimizer of \eqref{eq:mffunc}.

\medskip 
 
Consider first the case of  radially symmetric $U$. 
The sublevel set $\Omega_0$ is a union of disjoint annuli centered at the origin, 
\beq \Omega_0=\bigcup_{k=1}^K\mathcal A_k,\eeq
with 
\begin{equation} 
\mathcal A_k=\{x:\, r_{k,<}\leq |x|\leq  r_{k,>}\}, |\mathcal A_k|=\pi(r_{k,>}^2-r_{k,<}^2), \quad \quad\sum_{k=1}^K|\mathcal A_k|=\pi\ell.
\end{equation}
The complement $\Omega_0'=D(0,R)\setminus \Omega_0$ of $\Omega_0$ consists of annuli $\mathcal B_k$ with 
\begin{equation}
 \mathcal A_{k-1}\subset \mathcal B_k \subset \mathcal A_{k},\quad k=1,\cdots,K 
\end{equation}
with the convention that $\mathcal A_{0}=\emptyset$ and $\mathcal B_{1}=\emptyset$ if $r_{1,<}=0$.

If we fill $\Omega_0'=\cup_k \mathcal B_k$ with a two-dimensional continuous distribution of quasi-hole charges with density $2/\pi$ we obtain the same solution as before, Eq.\ \eqref{eq:cont charge}. Other solutions are obtained, however by distributing the total quasi-hole charge in $\mathcal B_k$ uniformly over any sub-annulus of $\mathcal B_k$. The sub-annulus can even shrink to a circle, $\mathcal C_k$, in which case we have a one-dimensional distribution. If $r_{1,<}=0$, we can shrink $\mathcal C_1$ to a point and have a quasi hole at the origin. 

The approximation by discrete distributions can in all cases carried out as before by covering the annuli/circles by a discrete grid.

\medskip

Also for general $U$ alternatives to \eqref{eq:cont charge} can be obtained making use of Newton's theorem. Let $D_1,\dots,D_n$ be any finite collection of disjoint disks contained in $D(0,R)\setminus \Omega_0$ such that $(N/\pi)|D_i|$ is a positive integer. Then a new solution is obtained by replacing the continuous distribution \eqref{eq:cont charge} within each disk $D_j$ by a discrete charge $q_j=(2/\pi)|D_j|$ placed at the center of the disk. The variational equation~\eqref{eq:var} holds with this new quasi-hole distribution because outside of $D_j$ and in particular on $\Omega_0$  the potential generated by the central charge is the same as the one generated by a continuous distribution on the disk, while inside the disk it is not smaller than before.

Using the \lq\lq cheese theorem\rq\rq\ \cite{LieLeb-72, LieSei-09} one can choose the disks so that they cover the whole of $D(0,R)\setminus \Omega_0$ up to an arbitrary small remainder. This is another way to approximate \eqref{eq:cont charge} by a discrete distribution of quasi-hole charges.


\bibliographystyle{siam}

\end{document}